\documentclass[11pt,reqno,a4paper]{article}
\usepackage{amssymb,amsmath}
\usepackage{amssymb,amsthm}
\usepackage{enumerate}
\usepackage{placeins}
\usepackage{float}
\restylefloat{table}
\usepackage{cite}
\newtheorem{Theorem}{Theorem}
\newtheorem{Proof}{Proof}

\newtheorem{Lemma}{Lemma}
\newtheorem{Corollary}{Corollary}
\newtheorem{Definition}{Definition}
\numberwithin{Theorem}{section}
\numberwithin{Lemma}{section}
\newtheorem{Example}{Example}
\numberwithin{Corollary}{section}
\numberwithin{Example}{section}
\numberwithin{Remark}{section}
\usepackage{amsmath}
\usepackage{amsfonts}
\usepackage{amssymb}
\date{}
\begin{document}
\title{\textbf{ On cyclic self-orthogonal codes over  $\mathbb{Z}_{p^m}$ }}
\author{Abhay Kumar Singh, Narendra Kumar 
\\\textit{Department of Applied Mathematics }
\\\textit{Indian School of Mines, Dhanbad-826004}\\}
\maketitle 
\footnote{\hspace{-.65cm} singh.ak.am@ismdhanbad.ac.in, narendrakumar9670@gmail.com}
\begin{abstract}
The purpose of this paper is to study the cyclic self orthogonal codes over  $\mathbb{Z}_{p^m}$. After providing the generator polynomial of cyclic self orthogonal codes over  $\mathbb{Z}_{p^m}$, we give the necessary and sufficient condition for the existence of non-trivial self orthogonal codes over $\mathbb{Z}_{p^m}$ . We have also provided the number of such codes of length $n$ over $\mathbb{Z}_{p^m}$ for any $ (p,n) = 1 $.  
\end{abstract}
\section{Introduction}
In recent time , self-orthogonal codes over finite rings have been studied extensively because of their close connection with other mathematical structure such as block design , lattices and modular forms . The complete structure of cyclic codes of odd length over $\mathbb{Z}_{4}$ have been discussed in series of papers [13,15,17]. In [2,3,4,6,12,14,18] , the cyclic codes of odd length over finite rings have been presented . Cyclic codes over $\mathbb{Z}_{p^e}$  of length $N$ where $p$ does not divide $N$ have been discussed in [9,13] . In [11] , this work was completed by studying cyclic codes of repeated roots over $\mathbb{Z}_{p^m}$ . The structure of self-dual and self-orthogonal codes over $\mathbb{Z}_{4}$ have been discussed in [1]. In [5] , Dougherty et studied the self-dual and self-orthogonal codes over $\mathbb{Z}_{2^m}$ . Hoffman has given excellent summary of self-dual codes over finite rings in [7] . The intersection of cyclic codes and orthogonal codes is called the cyclic self-orthogonal codes . Cyclic self-orthogonal codes over $\mathbb{Z}_{2^m}$ of odd length were investigated by Qian et al. in [10] and generator polynomial for such codes were given .
Generally, cyclic self-orthogonal codes finite chain ring can be studied in two folds : simple root cyclic self-orthogonal codes, if the code lengths are coprime  to the characteristic of ring ; otherwise we have so called repeated root cyclic self-orthogonal codes . In this paper we study simple root cyclic self-orthogonal codes over $\mathbb{Z}_{p^m}$ .
\par In this paper , we study simple root cyclic self-orthogonal codes over $\mathbb{Z}_{p^m}$ by using generator polynomial . The paper is organized as follows . In sec $2$ , we give some definition and results related with this work . Generator polynomial is given for the simple root self-orthogonal cyclic codes in sec $3$ . sec $4$ gives sufficient and necessary condition for the existence simple root cyclic self-orthogonal codes over $\mathbb{Z}_{p^m}$ . In sec $5$, cyclic self-dual codes over $\mathbb{Z}_{p^m}$ for the simple root are studied . Summery of this paper given in sec $6$ .
\section{Preliminaries}
\par  The ring $\mathbb{Z}_{p^m}$ is a local principal ideal ring with maximal ideal $\left\langle p \right\rangle$ . Each element $a\in\mathbb{Z}_{p^m}$ can be expressed uniquely in the form 
\begin{align*}
a = a_{0}+pa_{1}+p^2a_{2}+. . . . . +p^{m-1}a_{m-1}
\end{align*}
 where $a_{i}\in[0,1,2, . . . . . . .p-1]$ for $0\leq i\leq m-1$. Two polynomials are coprime if there exist $\lambda_{1}(x)$,$\lambda_{2}(x)\in\mathbb{Z}_{p^m}[x]$ such that $\lambda_{1}(x)f_{1}(x)$ + $\lambda_{2}(x)f_{2}(x)$ = $1$ . A polynomial $f(x)\in\mathbb{Z}_{p^m}[x]$ is called basic irreducible if its reduction modulo $p$, denoted by $\tilde{f}(x)$, is irreducible in $\mathbb{Z}_{p}[x]$.Two polynomials $f_{1}(x)$, $f_{2}(x)\in\mathbb{Z}_{p^m}[x]$ are coprime in $\mathbb{Z}_{p^m}[x]$ if and only if $\tilde{f}_{1}(x)$ , $\tilde{f}_{2}(x)$ are coprime in $\mathbb{Z}_{p}[x]$. By Hensel's lemma [$15, lemma 2.8$], f(x) can factor uniquely as a product of monic basic irreducible pairwise coprime polynomials. The Galois ring $GR(p^m,\nu)$ of characteristic $p^m$ and cardinality $p^{m\nu}$ is a finite chain ring of length $m$. The ring $GR(p^m ,\nu)$ is isomorphic to the residue class ring $\mathbb{Z}_{p^m}[x]/{\left\langle h(x)\right\rangle}$, where $h(x)$ is a monic basic irreducible polynomial of degree $\nu$ in $\mathbb{Z}_{p^m}[x]$.
\par A linear code $C$ of length $n$ over $\mathbb{Z}_{p^m}$ is a $\mathbb{Z}_{p^m}$-submodule of $\mathbb{Z}^n_{p^m}$, A linear code $C$ of length $n$ over  $\mathbb{Z}_{p^m}$ is called cyclic if it is invariant under the cyclic shift  operator $\tau$ :
\begin{align*}
\textbf{c} = (c_{0},c_{1}, . . . . . ,c_{n-1})\in C  \Rightarrow  \tau(\textbf{c}) = (c_{n-1},c_{0}, . . . . . ,c_{n-2})\in C.
\end{align*}
A codeword $\textbf{c} = (c_{0},c_{1}, . . . . . ,c_{n-1})\in C $ can be written as in polynomial form  $c(x)$ = $c_{0}+c_{1}x+ . . . . . +c_{n-1}x^{n-1}$. A code over $\mathbb{Z}_{p^m}$ is  cyclic if and only if it is ideal in the ring $\mathbb{Z}_{p^m}/{\left\langle {x^n}-1 \right\rangle}$. Throught this paper , we take the length $n$  is odd . We know that the ring $\mathbb{Z}_{p^m}/{\left\langle {x^n}-1 \right\rangle}$ is a principle ideal ring . For a monic divisor $f(x)$ of ${x^n}-1$ in $\mathbb{Z}_{p^m}[x]$, we denote $\hat{f}(x) = {x^n}-1/{f(x)}$.\par Given two n-tuples $\textbf{u} = (u_{0},u_{1},, . . . . ,u_{n-1})$ and $\textbf{v} = (v_{0},v_{1},, . . . . ,v_{n-1})\in\mathbb{Z}^n_{p^m}$ , their Euclidean inner product is defined as $\textbf{u.v} = u_{0}v_{0}+u_{1}v_{1}+ . . . . . +u_{n-1}v_{n-1}$ (evaluated in $\mathbb{Z}_{p^m}$) . For a linear code $C$ over $\mathbb{Z}_{p^m}$ of length $n$ ,the dual code of $C$ is defined as 
\begin{equation*}
 C^{\bot} = [\textbf{u} \in\mathbb{Z}^n_{p^m}|\textbf{u.v} = 0 for all \textbf{v}\in C].
\end{equation*}
A linear code $C$ over $\mathbb{Z}_{p^m}$ of length $n$ is called self-orthogonal if $ C \subseteq C^\bot $, it is called self-dual if $ C = C^\bot $ . Let $ f(x)  =   a_{\epsilon}x^{\epsilon}+a_{\epsilon-1}x^{\epsilon-1} + . . . . . +a_{0}$ be a monic polynomial in $\mathbb{Z}_{p^m}[x]$,where $a_{o}$ is a unit in $\mathbb{Z}_{p^m}$ . Define the reciprocal polynomial of $f(x)$ as $ f^{*}(x) = {a^{-1}_{0}}x^{deg(f(x))}f(x^{-1})$,that is , $ f^*(x) = x^{\epsilon}+a^{-1}_{0}a_{1}x^{\epsilon-1}+ . . . . . . +a^{-1}_{0}a_{\epsilon}$. Note that, $f^*(x)$ is also a monic polynomial in $\mathbb{Z}_{p^m}[x]$. If $f(x) = f^*(x) $, then $f(x)$ is called self-reciprocal over $\mathbb{Z}_{p^m}$ . Let $C$ be an ideal in $\mathbb{Z}_{p^m}[x]/{\left\langle x^{n}-1\right\rangle}$ . If $f(x)g(x) = 0$ in $\mathbb{Z}_{p^m}[x]/{\left\langle x^{n}-1\right\rangle}$, for all $ g(x)\in C $,then it is obvious that $f^*(x)$ must be in $ C^{\bot} [12]$. Following result gives generator polynomials for cyclic code $C$.
\begin{Theorem}
(see[6,12,14]) Let $C$ be a cyclic code over $\mathbb{Z}_{p^m}$ of length $n$.Then there exists a unique family of pairwise co-prime polynomials $G_{0}(x)$,$G_{1}(x)$, . . . ,$G_{m}(x)$ in $\mathbb{Z}_{p^m}[x]$ such that $G_{0}(x)G_{1}(x) . . . ,G_{m}(x) = {x^n}-1$ and $ C = \left\langle \hat{G}_{1}(x), p\hat{G}_{2}(x), . . . . , p^{m-1} \hat{G}_{m}(x)\right\rangle $. Moreover ,
$|C| =  p^{{\sum_{i=0}^{m-1}}(m-i)deg(G_{i+1})} $.
\end{Theorem}
\section{Generator polynomials of cyclic self-orthogonal codes over $\mathbb{Z}_{p^m}$}
\par  In this section , we first give an alternative generator for a cyclic code of length $n$ over $\mathbb{Z}_{p^m}$ , which depends on the unique factorization of ${x^n}+1$ in $\mathbb{Z}_{p^m}[x]$. Let ${x^n}-1$ can be unique factored into monic basic irreducible polynomials in $\mathbb{Z}_{p^m}[x]$ given below
\begin{equation*}
 x^n-1 = g_{1}(x)g_{2}(x). . . g_{r}(x).
\end{equation*}
 Set $ s_{0} = 0 $ and $ s_{m+1} = r $. Define $G_{i}(x) = g_{{s_i}+1}(x). . . g_{s_{i+1}}(x)$,where $s_{i}\leq s_{i+1}$ for $0\leq i\leq m$. Note that if $s_i=s_{i+1}$ we get $G_{i}(x)=1$. Then $G_{0}(x)G_{1}(x) . . . G_{m}(x) = {x^n}-1$ in $\mathbb{Z}_{p^m}[x]$.We know from Theorem $2.1$, a cyclic code $C$ of length n over $\mathbb{Z}_{p^m}$ is given by
\begin{equation*}
C = \left\langle\hat{G}_{1}(x),p\hat{G}_{2}(x),. . ., p^{m-1}\hat{G}_{m}(x)\right\rangle 
\end{equation*}
where $\hat{G}_{i}(x) = ({x^n}-1)/ G_{i}(x)$ for $0\leq i\leq m.$
Using Hensel's lemma,we can factored ${x^n}+1$ in $\mathbb{Z}_{p^m}[x]$ into monic basic irreducible polynomials
\begin{equation}
{x^n}+1 = f_{1}(x)f_{2}(x) . . . f_{r}(x),
\end{equation}
where $ \tilde{f}_{i}(x) = \tilde{g}_{i}(x)$for $1\leq i\leq r $. Let $s_{i},i=0,1,. . . ,{m+1}$, be defined as above . Define $F_{i}(x)=f_{{s}_{i+1}+1}. . . . f_{s_{i+2}}$ for $0\leq i\leq {m-1}$, and $F_{m}(x) = f_{1}(x). . . f_{s_{1}}(x)$.
Then, $\tilde{F}_{i}(x) = \tilde{G}_{i+1}$ for $0\leq i\leq {m-1}$ and $\tilde{F}_{m}(x) = \tilde{G}_{0}(x)$ .Hence $F_{0}(x)F_{1}(x). . . F_{m}(x) = {x^n}+1$ in $\mathbb{Z}_{p^m}[x]$. Define $\hat{F}_{i}(x) = {x^n}+1 / F_{i}(x)$ for $0\leq i\leq m$.
\begin{Theorem} 
 Using above notations,let C be cyclic code of the length n over $\mathbb{Z}_{p^m}$ is of the form  $C = \left\langle\hat{G}_{1}(x),p\hat{G}_{2}(x),. . ., p^{m-1}\hat{G}_{m}(x)\right\rangle $ . Then $C = \left\langle\prod_{j=0}^m{F _{j}(x)}^j\right\rangle$ and $|C| = p^{{\sum_{j=0}^m}(m-j)deg(F_{j})} $.
\end{Theorem}
\begin{Proof}
We take $F(x) = \prod_{j=0}^m{F _{j}(x)}^j$ .We know that $\hat{G}_{i}(x)$ and $ G_{i}(x)$ are co-prime in  $\mathbb{Z}_{p^m}[x]$, there exists $u_{i}(x) , v_{i}(x)\in \mathbb{Z}_{p^m}[x] $ such that $u_{i}(x)G_{i}(x) + v_{i}(x) \hat{G}_{i}(x)= 1 $.Multiplying both sides by $  \hat{G}_{i}(x)$, we find that $  \hat{G}_{i}(x) =v_{i}(x)\hat{G}_{i}(x)^2 $ in  $\mathbb{Z}_{p^m}[x]/\left\langle{x^n}-1\right\rangle $.Hence,we have 
\begin{equation*} 
\hat{G}_{i}(x) =  v_{i}(x)\hat{G}_{i}(x)^2 = v_{i}(x)^2\hat{G}_{i}(x)^3 = . . . = v_{i}(x)^{m-1}\hat{G}_{i}(x)^m .
\end{equation*}
For each $0\leq j\leq m-1$, since $ \tilde{G}_{j+1}(x) = \tilde{F}_{j}(x)$, this means that  $ \hat{G}_{j+1}(x) =  \hat{F}_{j}(x) + pu_{j}(x) $, for some $ u_{j}(x)\in\mathbb{Z}_{p^m}[x]$.Hence for each $0\leq j \leq m-1$, we calculate
\begin{align*}
p^j \hat{G}_{j+1}(x) &= p^jv_{j+1}(x)^{m-1}[\hat{F}_{j}(x) + pu_{j}(x)]^m\\
&= ({x^n}+1)^jv_{j+1}(x)^{m-1}[\hat{F}_{j}(x)+({x^n}+1)u_{j}(x)]^m\\
&= [\hat{F}_{j}(x)F_{j}(x)]^jv_{j+1}(x)^{m-1}[\hat{F}_{j}(x)+\hat{F}_{j}(x)F_{j}(x)u_{j}(x)]^m\\
&= v_{j+1}(x)^{m-1}[1+u_{j}(x)F_{j}(x)]^m{{F}_{j}(x)}^j{\hat{F}_{j}(x)}^{j+m}
\end{align*}
It means that $ p^j\hat{G}_{j+1}(x)\in\left\langle F(x)\right\rangle$ for each $0\leq j \leq m-1$ .Hence $C\subseteq \left\langle F(x)\right\rangle$.
On the contrary, we know that $ {G_{i}(x)}^m $ and $  {\hat{G}_{i}(x)}^m $ are relatively co-prime, so, there exists $\alpha_{i}(x),\beta_{i}(x)\in\mathbb{Z}_{p^m}[x]$ such that $ \alpha_{i}(x){G_{i}(x)}^m + \beta_{i}(x) {\hat{G}_{i}(x)}^m = 1 $.This fallows that
\begin{equation}
\prod_{i=1}^m[\alpha_{i}(x){G_{i}(x)}^m + \beta_{i}(x) {\hat{G}_{i}(x)}^m ]= 1.       
\end{equation}
Note that for distinct ${i,j}\in {1,2, . . . ,m },\hat{G}_{i}(x)\hat{G}_{j}(x) = 0 $ in  $ \mathbb{Z}_{p^m}[x]/ \left\langle{x^n}-1\right\rangle $. Expanding the left- hand side of $(2)$, we get that there exists $a_{0}(x),a_{1}(x), . . . ,a_{m}(x)\in\mathbb{Z}_{p^m}[x]$ such that 
\begin{equation*}
a_{0}(x)[G_{1}(x)G_{2}(x) . . . G_{m}(x)]^m + a_{1}(x)[\hat{G}_{1}(x)G_{2}(x) . . . . .  G_{m}(x)]^m + . . . 
\end{equation*}
\begin{equation}
+ a_{m}(x)[{G}_{1}(x)G_{2}(x) . . . . .G_{m-1}(x) \hat{G}_{m}(x)]^m = 1.
\end{equation}
Multiplying both sides of $(3)$ by $F(x)$,we have 
\begin{equation}
F(x) = b_{0}(x){\hat{G}_{0}(x)}^mF(x) +  b_{1}(x){\hat{G}_{1}(x)}^mF(x) + b_{2}(x){\hat{G}_{2}(x)}^mF(x) + . . . . .  +  b_{m}(x){\hat{G}_{m}(x)}^mF(x)
\end{equation}
for some $b_{0}(x),b_{1}(x),. . . ,b_{m}(x)\in \mathbb{Z}_{p^m}[x]$. Since $\tilde{G}_{i+1}(x) = \tilde{F}_{i}(x)$ for $0\leq i \leq m-1$, this means that ${F}_{i}(x) = G_{i+1}(x) + pV_{i+1}(x)$, for some $V_{i+1}(x)\in  \mathbb{Z}_{p^m}[x] $. For each $0\leq i \leq m-1$,calculating in $ \mathbb{Z}_{p^m}[x]/ \left\langle{x^n}-1\right\rangle $, we have 
\begin{equation*}
{\hat{G}_{i+1}(x)}^{i+1}F(x) = {\hat{G}_{i+1}(x)}^{i+1}[ G_{i+1}(x) + pV_{i+1}(x)]^i\prod_{0\leq j \leq m,j\neq i} {F_{j}(x)}^j
\end{equation*}
\begin{equation}
~~~~~~~~~~~~~~~~~= p^i{\hat{G}_{i+1}(x)}[V_{i+1}(x)\hat{G}_{i+1}(x)]^i\prod_{0\leq j \leq m,j\neq i} {F_{j}(x)}^j.
\end{equation}
Note that ${\hat{G}_{0}(x)}^{m}F(x) = {\hat{G}_{0}(x)}^{m}[ G_{0}(x) + pV_{0}(x)]^m \prod_{j=0}^{m-1}{F_{j}(x)}^j = 0 $ in  $ \mathbb{Z}_{p^m}[x]/ \left\langle{x^n}-1\right\rangle $. From $ (4) $ and $ (5) $, we obtain that $ F(x) = k_{1}(x)\hat{G}_{1}(x)  +  k_{2}(x).{p}\hat{G}_{2}(x)  + . . . . +  k_{m}(x).{p}^{m-1}\hat{G}_{m}(x)$ , for some $ k_{1}(x),k_{2}(x), . . . ,k_{m}(x)\in \mathbb{Z}_{p^m}[x] $. This gives that $ F(x)\in C $ and $\left\langle F(x)\right\rangle \subseteq C $ . Thus, $ C = \left\langle \prod_{j=0}^m{F _{j}(x)}^j \right\rangle $.\\
For cardinality of cyclic code $ C $\\
If each $ F_{j}\neq 1 (0\leq j\neq {m-1})$. Then they are pairwise, co-prime and thus we have
\begin{equation*}
C = (\hat{F}_{0}(x))\bigoplus (P \hat{F}_{1}(x))\bigoplus (P^2 \hat{F}_{2}(x))\bigoplus . . . . \bigoplus (P^{m-1} \hat{F}_{m-1}(x)).
\end{equation*}
Thus ,
\begin{align*}
 |C| &= |( \hat{F}_{0}(x))|(P \hat{F}_{1}(x))|. . . . . |(P^{m-1} \hat{F}_{m-1}(x))|\\
 &= p^{{m}(n-deg\hat{F}_{0}(x))}p^{{m-1}(n-deg\hat{F}_{1}(x))}. . . . p^{(n-deg\hat{F}_{m}(x))}\\
 & = p^{\sum _{j=0}^m(m-j)deg F_{j}} .
\end{align*}
\end{Proof}
\par From the above theorem , we see  that generator polynomials of cyclic codes of length n over $ \mathbb{Z}_{p^m} $ depend on the factorization of ${x^n}+1$ over  $ \mathbb{Z}_{p^m}$. Let ${x^n}+1$ can be factored into monic basic irreducible polynomials in $\mathbb{Z}_{p^m}[x]$ , written as 
\begin{equation}
{x^n}+1 = f_{1}(x) . . . . f_{s}(x)h_{1}(x)h_{1}^*(x) . . . . h_{t}(x)h_{t}^*(x)
\end{equation}
where $f_{i}(x)(1\leq i\leq s)$ are basic irreducible self-reciprocal polynomials in $ \mathbb{Z}_{p^m}[x] $ and $ h_{j}(x) $ and $ h^*_{j}(x) (1\leq j\leq t)$ are basic irreducible reciprocal polynomial pairs in $ \mathbb{Z}_{p^m}[x] $. The fallowing result gives the generator polynomials of dual codes of cyclic over $ \mathbb{Z}_{p^m} $.
\begin{Lemma}
Let ${x^n}+1$ have the unique factorization over  $ \mathbb{Z}_{p^m} $ as given in $(6)$. Let C be a cyclic code over  $ \mathbb{Z}_{p^m} $ of length n with generator polynomial
\begin{equation}
G(x) = {f_{1}(x)}^{\ell_{1}}. . . . . {f_{s}(x)}^{\ell_{s}}{h_{1}(x)}^{\kappa_{1}}{h_{1}^*(x)}^{\lambda_{1}}. . . . . {h_{t}(x)}^{\kappa_{t}}{h_{t}^*(x)}^{\lambda_{t}}
\end{equation}
where $ 0 \leq \ell_{i},\kappa_{j},\lambda_{j}\leq m $ for each $ 1\leq i\leq s $ and $ 1\leq j\leq t $.Then its dual code $ C^\bot $ has generator polynomial 
\begin{equation}
 G^\bot(x) = {f_{1}(x)}^{m-\ell_{1}} . . . . {f_{s}(x)}^{m-\ell_{s}}{h_{1}(x)}^{m-\lambda_{1}}{h_{1}^*(x)}^{m-\kappa_{1}} . . . . . {h_{t}(x)}^{m-\lambda_{t}}{h_{t}^*(x)}^{m-\kappa_{t}} .
\end{equation}
\end{Lemma}
\begin{proof}
Let D be the cyclic code over $ \mathbb{Z}_{p^m} $ of length n with generator polynomial as given in $(8)$. We know that $ G(x)(G^\bot(x))^* = f_{1}(x)^m . . . . f_{s}(x)^mh_{1}(x)^mh_{1}^*(x)^m . . . . h_{t}(x)^mh_{t}^*(x)^m  = ({x^n}+1)^m = 0 $ in $ \mathbb{Z}_{p^m}[x]/ \left\langle{x^n}-1\right\rangle $ . This gives that $ D\subseteq C^\bot $ . Also ,$ |D| = |C^\bot| = p^{\sum _{i=0}^s(\ell_{i}) + \sum_{i=1}^t(\kappa_{i}+\lambda_{i})}$. Therefore , $ D = C^\bot $.
\end{proof}
\begin{Theorem}
Let ${x^n}+1$ have the unique factorization over $ \mathbb{Z}_{p^m} $ as given in $(6)$ . Let C be a cyclic code of length $n$ over $ \mathbb{Z}_{p^m} $ with generator polynomial as in $(7)$ . Then C is self-orthogonal if and only if $\lceil m/2 \rceil \leq \ell_{i}\leq m $ for each $ 1\leq i \leq s $ and $ m \leq {\kappa_{j} + \lambda_{j}}\leq 2m $ for each $ 1\leq j \leq t $ .
\end{Theorem}
\begin{proof}
Let C be a cyclic code of length $n$ over $ \mathbb{Z}_{p^m} $ have generator polynomial $G(x)$ as in $(7)$. Then , lemma $(3.2)$ implies that $ C^\bot $ has generator polynomial $ G^\bot(x) $ which is given by $(8 )$ . The code C is self - orthogonal if and only if $ C\subseteq C^\bot $; if and only if $ G^\bot(x) $ divides $ G(x)$ .This means that $\ell_{i}\geq {m-\ell_{i}}$ for each $ 1\leq i \leq s $, $\lambda_{j}\geq {m-\kappa_{j}}$ and $\kappa_{j}\geq {m-\lambda_{j}}$ for each $ 1\leq j \leq t $ . As we know that $ 0 \leq \ell_{i},\kappa_{j},\lambda_{j}\leq m $ for each $ 1\leq i\leq s $ and $ 1\leq j\leq t $ . Thus , a cyclic code C is self-orthogonal if and only if  $\lceil m/2 \rceil \leq \ell_{i}\leq m $ for each $ 1\leq i \leq s $ and $ m \leq {\kappa_{j} + \lambda_{j}}\leq 2m $ for each $ 1 \leq j \leq t$.
\end{proof}
\par The equation $(7)$, if we take $ \ell_{i} = m $ for each $ 1\leq i \leq s$ 
and $ \kappa_{j} = \lambda_{j} = m $ for each  $ 1\leq j\leq t $,then we find that $ C = \left\langle 0 \right\rangle $ ;if we take  $ \ell_{i} = \lceil m/2 \rceil $
for each $ 1\leq i \leq s $ and $ \kappa_{j} = \lambda_{j} = \lceil m/2 \rceil $
for each $ 1\leq j\leq t $,then $ C = \left\langle p^{\lceil m/2 \rceil} \right\rangle$ .By theorem $ 3.3 $, these two codes are both self-orthogonal . Thus, there exist at least two cyclic self-orthogonal codes over $ \mathbb{Z}_{p^m} $ for any odd length $n$ . Theorem $3.3$ shows that cyclic self-orthogonal codes over $ \mathbb{Z}_{p^m} $ of length $n$ can be determined by monic basic irreducible divisors of ${x^n}+1$ over  $ \mathbb{Z}_{p^m} $ and their exponents.
\begin{Definition}
Let $n$ be an odd integer. Define $ \gamma(n) $ to be the number of basic irreducible self-reciprocal polynomials in the factorization of ${x^n}+1$ in $ \mathbb{Z}_{p^m}[x] $,and $ \delta(n) $ be the number of basic irreducible reciprocal polynomial pairs in the factorization of $ {x^n}+1 $ in $ \mathbb{Z}_{p^m}[x] $.
\end{Definition}
\par According to Theorem $3.3$ , for a cyclic self-orthogonal code of length $n$ over $ \mathbb{Z}_{p^m} $ , the exponent of basic irreducible self-reciprocal polynomials in the factorization of ${x^n}+1$ in $ \mathbb{Z}_{p^m}[x] $ may be some integer $\epsilon$ in the range $ \lceil m/2 \rceil \leq \epsilon \leq m $, while the sum of the exponents of basic irreducible  reciprocal polynomial pairs may be $ m,. . . . ,{2m-1}$ or $ 2m $. Now we get next result .
\begin{Theorem}
Let ${x^n}+1$ have the unique factorization over $ \mathbb{Z}_{p^m}[x] $ as in $(6)$ , where $ s = \gamma(n) $ and $ t = \delta(n) $. Then the number of cyclic  self-orthogonal code of length $n$ over $ \mathbb{Z}_{p^m} $ is given by
\begin{equation}
\bigg(m-{\lceil \dfrac{m}{2} \rceil} + 1\bigg)^{\gamma(n)}{\left(\dfrac{(m+1)(m+2)}{2}\right)^{\delta(n)}}
\end{equation}
\end{Theorem}
\begin{Example}
Consider cyclic self-orthogonal codes over $ \mathbb{Z}_{8} $ of length $7$. In $ \mathbb{Z}_{8}[x] $,${x^7}+1 = f_{1}(x)f_{2}(x)f_{3}(x)$ where 
\begin{align*}
f_{1}(x) = x+1  ,    f_{2}(x) = x^3+5x^2+2x+1  ,    f_{3}(x) = x^3+2x^2+5x+1 .
\end{align*}
Here $  f_{1}(x) $ is  self-reciprocal over $ \mathbb{Z}_{8} $,and  $ f_{2}(x)$ and $ f_{3}(x) $ are a reciprocal polynomial pair over $ \mathbb{Z}_{8} $.There exist $20$ cyclic self-orthogonal code over $ \mathbb{Z}_{8} $ of length $7$ . Their generator polynomials are\\ 
$ (1)~ f_{1}(x)^3f_{2}(x)^3f_{3}(x)^3  = 0 $\\
$ (2)~ f_{1}(x)^3f_{2}(x)^3f_{3}(x)^2  =  4x^4+4x^2+4x+4 $\\
$ (3)~ f_{1}(x)^3f_{2}(x)^3f_{3}(x)^1  =  2x^8+4x^6+4x^5+6x^4+4x^3+2 $\\
$ (4)~ f_{1}(x)^3f_{2}(x)^3  = x^{12}+2x^{11}+x^{10}+5x^9+6x^8+2x^6+7x^5+7x^4+3x^3+x+1 $\\
$ (5)~f_{1}(x)^3f_{2}(x)^2f_{3}(x)^3  =  4x^4+4x^3+4x^2+4 $\\
$ (6)~ f_{1}(x)^3f_{2}(x)^1f_{3}(x)^3  =  2x^8+4x^7+6x^6+4x^5+6x^4+4x^3+4x^2+2 $\\
$ (7)~ f_{1}(x)^3f_{3}(x)^3 = x^{12}+x^{11}+3x^9+3x^8+5x^7+4x^6+4x^5+6x^4+5x^3+5x^2+2x+1 $\\
$ (8)~  f_{1}(x)^3f_{2}(x)^2f_{3}(x)^2  =  4x+4 $\\
$ (9)~  f_{1}(x)^3f_{2}(x)^2f_{3}(x)^1  =  2x^5+6x^4+2x^3+4x^2+2 $\\
$ (10)~ f_{1}(x)^3f_{2}(x)^1f_{3}(x)^2  =  2x^5+4x^3+2x^2+6x+2 $\\
$ (11)~ f_{1}(x)^2f_{2}(x)^3f_{3}(x)^3  =  4x^6+4x^5+4x^4+4x^3+4x^2+4x+4 $\\
$ (12)~  f_{1}(x)^2f_{2}(x)^3f_{3}(x)^2  =  4x^3+4x^2+4 $\\
$ (13)~  f_{1}(x)^2f_{2}(x)^3f_{3}(x)^1  =  2x^7+6x^6+6x^5+6x^4+2x+2 $\\
$ (14)~  f_{1}(x)^2f_{2}(x)^3  =  x^{11}+x^{10}+5x^8+x^7+7x^6+3x^5+4x^4+3x^3+1 $\\
$ (15)~  f_{1}(x)^2f_{2}(x)^2f_{3}(x)^3  =  4x^3+4x+4 $\\
$ (16)~  f_{1}(x)^2f_{2}(x)^1f_{3}(x)^3  =  2x^7+2x^6+4x^5+6x^3+6x^2+6x+2 $\\
$ (17)~  f_{1}(x)^2f_{3}(x)^3  =  x^{11}+3x^8+5x^6+7x^5+5x^4+x^3+4x^2+x+1 $\\
$ (18)~  f_{1}(x)^2f_{2}(x)^2f_{3}(x)^2  =  4 $\\
$ (19)~  f_{1}(x)^2f_{2}(x)^1f_{3}(x)^2  =  2x^4+6x^3+6x^2+4x+2 $\\
$ (20)~  f_{1}(x)^2f_{2}(x)^2f_{3}(x)^1  =  2x^4+6x^2+2 $
\end{Example}
\section{ Existence of cyclic self-orthogonal codes over $ \mathbb{Z}_{p^m} $ }
\par In this section , we examine the conditions for the existence of cyclic self-orthogonal codes over $ \mathbb{Z}_{p^m} $. It fallows from Theorem $3.3$ that there are at least two cyclic self-orthogonal codes $ \left\langle 0 \right\rangle $ and$ \left\langle p^{\lceil m/2 \rceil}\right\rangle $ over $ \mathbb{Z}_{p^m} $ for any odd length .  
\begin{Definition}
For any odd length $n$,a cyclic self-orthogonal code over $ \mathbb{Z}_{p^m} $ is said to be trivial cyclic self-orthogonal code if it is contained in $ \left\langle p^{\lceil m/2 \rceil}\right\rangle $ . Otherwise , the code is called the non trivial cyclic self-orthogonal code .
\end{Definition}
\par There are two cases \\
Case $(1)$.
Let us suppose that there does not exist basic reciprocal polynomial pairs in the factorization of  ${x^n}+1$ in  $ \mathbb{Z}_{p^m}[x] $ $i.e.$ $ \delta(n) = 0 $ .   Then, by Theorem $3.3$, a cyclic self-orthogonal code C over $ \mathbb{Z}_{p^m} $  of length $n$ has generator polynomial  $ G(x) = {f_{1}(x)}^{\ell_{1}}. . . . . {f_{\gamma(n)}(x)}^{\ell_{\gamma(n)}} $ ,where $ \lceil m/2 \rceil\leq \ell_{i} \leq m $ for each $ 0\leq i\leq \gamma(n) $.Thus, we have
\begin{align*}
 G(x) &= ({x^n}+1)^{\lceil m/2 \rceil}f_{1}(x)^{\ell_{1}-\lceil m/2 \rceil} . . . . f_{\gamma(n)}(x)^{\ell_{\gamma(n)}-\lceil m/2 \rceil}\\
&= p^{\lceil m/2 \rceil}f_{1}(x)^{\ell_{1}-{\lceil m/2 \rceil}} . . . . f_{\gamma(n)}(x)^{\ell_{\gamma(n)}-{\lceil m/2 \rceil}}.
\end{align*}
It fallows that $ C = \left\langle G(x)\right\rangle $ is contained in $ \left\langle p^{\lceil m/2 \rceil}\right\rangle $. Therefore , $C$ is a trivial cyclic self-orthogonal code over $ \mathbb{Z}_{p^m} $.\\
Case $(2)$.
suppose that there exist basic irreducible reciprocal polynomial pairs in the factorization of  ${x^n}+1$ in  $ \mathbb{Z}_{p^m}[x] $, $i.e.$ $ \delta(n)\neq 0 $ . Then, by using Theorem $3.3$ , the cyclic code  over $ \mathbb{Z}_{p^m} $ of length $n$ with generator polynomial\\
\begin{equation*}
G(x) =  f_{1}(x)^m . . . . f_{\gamma(n)}(x)^mh_{1}(x)^m . . . . h_{\delta(n)}(x)^m.
\end{equation*}
is self-orthogonal. Thus above code is not of the form $ \left\langle p^{\lceil m/2 \rceil}\right\rangle $ . Hence , it is a nontrivial cyclic self-orthogonal code over $ \mathbb{Z}_{p^m} $ . Thus , we have obtained the fallowing result.
\begin{Lemma}
For any odd integer $n$, non trivial cyclic self-orthogonal code over $\mathbb{Z}_{p^m} $ of length  $n$ exist if and only if the number $\delta(n)$ of  basic irreducible reciprocal polynomial pairs in the in the factorization of  ${x^n}+1$ in  $ \mathbb{Z}_{p^m}[x] $ is non zero .
\end{Lemma}
\begin{Theorem}
For any odd integer $n$,non trivial cyclic self-orthogonal code over $\mathbb{Z}_{p^m}$ of length $n$ exist if and only if ${p^{i}\neq{-1}}(\textnormal{ mod}~ n )$ for any positive integer ${i}$.
\end{Theorem}
\begin{proof}
Let $\gamma$ be the multiplicative order of $ p\textnormal {mod}~(n)$ . Then , there exists a primitive $n$th root of unity in $GR(p^m,\nu)$. Let $\nu_{i}$ be size of the $p$-cyclotomic coset modulo $n$ containing $i$ for each $ i\in{ 0,1,2,. . . . ,n-1}$ . Let f(x) be the a monic basic irreducible divisor of  ${x^n}+1$ in $ \mathbb{Z}_{p^m}[x] $ . Then f(x) must have root $-{\epsilon}^i $ for some $ i\in {0,1,. . . . . ,n-1}$ . This fallows that set $(-\epsilon^j | j = i,pi,. . . . ,p^{\nu_{i}-1}i)$ have all distinct roots of $f(x)$ .The reciprocal polynomial of $f(x)$ is also monic basic irreducible in  $ \mathbb{Z}_{p^m}[x] $ and has all distinct roots  $(-\epsilon^{-j} | j = i,pi,. . . . ,p^{\nu_{i}-1}i)$ . If basic reciprocal polynomial pairs $\delta(n)$ in the factorization of  ${x^n}+1$ in  $ \mathbb{Z}_{p^m}[x] $ is non zero if and only if there exist monic basic irreducible divisors $ f(x)$ and $ f(x)^* $ in the factorization of  ${x^n}+1$ over $ \mathbb{Z}_{p^m} $ such that $  f(x)\neq  f(x)^* $ ;  if and only if $ i $ and $n-i$ are not in the same cyclotomic coset .This implies that  ${p^{i}\neq{-1}}(\textnormal{ mod}~ n )$ for any positive integer ${i}$.This conditions is necessary and sufficient for the existance of non-trivial cyclic self-orthogonal code over $\mathbb{Z}_{p^m}$ .
\end{proof}
\par Next,the next corollary gives the number of trivial and non-trivial cyclic self-orthogonal code over $\mathbb{Z}_{p^m}$ for given length . Let $C$ be a cyclic code of length $n$ over $\mathbb{Z}_{p^m}$ with generator polynomial given in $(7)$. Combining Theorem $3.3$ and Definition $4.1$ , we have C is trivial self-orthogonal if and only if $C \subseteq \left\langle p^{\lceil m/2\rceil}\right\rangle $,$ \lceil m/2\rceil \leq \ell_{i} \leq m $ for each $ 1 \leq i \leq \gamma(n)$ and $ \lceil m \rceil \leq {\kappa_{j}+\lambda_{j}}\leq 2m $ for each $ 1\leq j \leq \delta(n)$,which means that $ \lceil m/2 \rceil \leq \ell_{i}\leq m$ for each $ 1 \leq i \leq \gamma(n) $ and $ \lceil m/2 \rceil \leq {\kappa_{j},\lambda_{j}}\leq m $ for each $ 1\leq j \leq \delta(n)$.
\begin{Corollary}
Let ${x^n}+1$ have the unique factorization over $\mathbb{Z}_{p^m}$ as given in $(6)$,where $ s = \gamma(n)$ and $ t = \delta(n) $\\
$(1)$~ The number of trivial cyclic self-orthogonal codes over$ \mathbb{Z}_{p^m}$ of length $n$ is
\begin{equation}
\left(m-{\lceil\dfrac{m}{2}\rceil} + 1\right)^{\gamma(n)+2\delta(n)}
\end{equation}
$(2)$~ The number of nontrivial cyclic self-orthogonal codes over $ \mathbb{Z}_{p^m}$ of length $n$ is
\begin{equation}
\left(m-{\lceil \dfrac{m}{2} \rceil} + 1\right)^{\gamma(n)}\left[\left(\dfrac{(m+1)(m+2)}{2}\right)^{\delta(n)}-\left(m-{\lceil \dfrac{m}{2}\rceil} + 1\right)^{2\delta(n)}\right].
\end{equation}
\end{Corollary}
\section{ Cyclic self-dual codes over $\mathbb{Z}_{p^m}$ }
\par In this section ,we explain cyclic self-dual codes over $\mathbb{Z}_{p^m}$ for any odd length . We know that ${x^n}+1$ can factored into monic basic irreducible polynomials in $\mathbb{Z}_{p^m}[x]$ written as
\begin{equation}
{x^n}+1 = f_{1}(x) . . . . f_{\gamma(n)}(x)h_{1}(x)h_{1}^*(x) . . . . h_{\delta(n)}(x)h_{\delta(n)}^*(x)
\end{equation}
where $ f_{i}(x)(1\leq i\leq \gamma(n))$ are basic irreducible self-reciprocal polynomials in $\mathbb{Z}_{p^m}[x]$ and $ h_{j}(x) $ and $ h_{j}^*(x) (1\leq j\leq \delta(n))$ are basic irreducible reciprocal polynomial pairs in $ \mathbb{Z}_{p^m}[x] $ . Let C be a cyclic code of length $n$ over  $ \mathbb{Z}_{p^m} $ with generator polynomial 
\begin{equation}
G(x) = {f_{1}(x)}^{\ell_{1}}. . . . . {f_{\gamma(n)}(x)}^{\ell_{\gamma(n)}}{h_{1}(x)}^{\kappa_{1}}{h_{1}^*(x)}^{\lambda_{1}}. . . . . {h_{\delta(n)}(x)}^{\kappa_{\delta(n)}}{h_{\delta(n)}^*(x)}^{\lambda_{\delta(n)}}
\end{equation}
where $ 0 \leq {\ell_{i},\kappa_{j},\lambda_{j}}\leq m $ for each $ 1\leq i\leq \gamma(n)$ and $ 1\leq j\leq \delta(n) $ . Then $ C^{\bot} $ has generator polynomial
\begin{equation}
 G^{\bot}(x) = {f_{1}(x)}^{m-\ell_{1}} . . . . {f_{\gamma(n)}(x)}^{m-\ell_{\gamma(n)}}{h_{1}(x)}^{m-\lambda_{1}}{h_{1}^*(x)}^{m-\kappa_{1}} . . . . . {h_{\delta(n)}(x)}^{m-\lambda_{\delta(n)}}{h_{\delta(n)}^*(x)}^{m-\kappa_{\delta(n)}}.
\end{equation}
The code C is self-dual if and only if $ G(x) = G^{\bot}(x) $,if and only if $ 2\ell_{i} = m $ for each $1\leq i\leq \gamma(n)$ and $ {\kappa_{j}+\lambda_{j}} = m $ for each $1\leq j\leq \delta(n)$ . Thus , we have the following result.\\
\begin{Theorem}
For any odd length $n$, cyclic self-dual codes over $ \mathbb{Z}_{p^m} $ exist only if $m$ is even . In the case when $m$ is even , let  ${x^n}+1$ be factored in  $ \mathbb{Z}_{p^m}[x] $ as in $(12)$ .A cyclic code over $ \mathbb{Z}_{p^m} $ of length $n$ is self-dual if and only if its generator polynomial has the form 
\begin{equation}
f_{1}(x)^{m/2}. . . . . f_{\gamma(n)}(x)^{m/2}h_{1}(x)^{\kappa_{1}}h_{1}^*(x)^{m-\kappa_{1}}. . . . . h_{\delta(n)}(x)^{\kappa_{\delta(n)}}h_{\delta(n)}^*(x)^{m-\kappa_{\delta(n)}}
\end{equation}
where $ 0\leq \kappa_{i}\leq m $ for each $1\leq j\leq \delta(n)$. 
\end{Theorem}
\par Note that the exponents of the polynomials of each irreducible reciprocal polynomial pair should sum upto $ m $ . The number of choices of each $ \kappa_{i}$ is exactly $ m+1 $  . Thus ,we get following result.\\
\begin{Corollary} 
Assume that $m$ is even . Let ${x^n}+1$ be factorized in $ \mathbb{Z}_{p^m}[x] $ as in $(12)$ . Then the number of cyclic self-dual codes of length $n$ over $ \mathbb{Z}_{p^m} $ is $ (m+1)^{\delta(n)} $ . Moreover,the code $\left\langle p^m\right\rangle$ is the unique trivial cyclic self-dual code .
\end{Corollary} 
\begin{Definition}
Euclidean weight of an element $e$ in  $ \mathbb{Z}_{p^m} $ is defined as min[${e^2},(p^m-e)^2$] For an $n$-tuple $ c \in { \mathbb{Z}^n_{p^m}} $, the Euclidean weight $ \omega_{E}(c)$ of $c$ is defined as the rational sum of the Euclidean weights of all its components . The minimum Euclidean weight $ d_{E}(C)$ of a linear code C over  $ \mathbb{Z}_{p^m} $ is the smallest Euclidean weight among all non-zero codewords of C . A self-dual over $ \mathbb{Z}_{p^m} $ is called Type $II$ if all codewords have Euclidean weights a multiple of $p^{m+1}$, otherwise it is called Type $I$.
\end{Definition}
\begin{Theorem}
For any odd length $n$, every cyclic self-dual code over $ \mathbb{Z}_{p^m} $ is of Type ${I}$.\
\end{Theorem}
\begin{proof}
Let C be the cyclic self-dual code over $ \mathbb{Z}_{p^m} $ of length $n$ with generator polynomial $ G(x) $ as in $(15)$.If C be the trivial  self-dual code $\left\langle p^{m/2}\right\rangle $ ,then it is obvious that C is of Type $I$. Hence , we only we need to consider the case that C is non trivial . In this case , $ \delta(n)\neq {0} $. For each $1\leq i \leq \delta(n)$, note that $ \kappa_{i} $ and ${m-\kappa_{i}}$ must be in the range $ m/2 $ and $m$. So ,we can have that the generator polynomial $G(x)$ as
\begin{equation}
G(x) = {f_{1}(x)}^{m/2}. . . . .{ f_{\gamma(n)}(x)}^{m/2}h_{1}(x)^{\kappa_{1}}{h_{1}^*(x)}^{m-\kappa_{1}}. . . . . {h_{\delta(n)}(x)}^{\kappa_{\delta(n)}}{h_{\delta(n)}^*(x)}^{m-\kappa_{\delta(n)}}\\
\end{equation}
where $ m/2 \leq \kappa_{i}\leq m $ for each $1\leq j \leq \delta (n)$. Then
\begin{align*}
 F(x)& ={ f_{1}(x)}^{m/2}. . . . . {f_{\gamma(n)}(x)}^{m/2}{h_{1}(x)}^{m}{h_{1}^*(x)}^{m/2}. . . . . {h_{\delta(n)}(x)}^{m}{h_{\delta(n)}^*(x)}^{m/2}\\
 & = ({x^n}+1)^{m/2}{h_{1}(x)}^{m/2}{h_{2}(x)}^{m/2}. . . . . {h_{\delta(n)}(x)}^{m/2}\\
 & = p^{m/2}[h_{1}(x)h_{2}(x). . . . . {h}_{\delta(n)}(x)]^{m/2}\\
\end{align*}
must be in C . Write $ H(x) = [h_{1}(x)h_{2}(x). . . . . h_{\delta(n)}(x)]^{m/2}$. By the Euclidean algorithm for finite commutative rings ,there exist $s(x),r(x)\in  \mathbb{Z}_{p^m}[x] $ such that $ H(x) = (x^n-1)s(x)+r(x)$,where $deg(r(x))<n $ . Then , $\tilde H(x) = (x^n-1)\tilde s(x)+\tilde r(x)$ in $\mathbb{F}_{p}[x]$. Since $\tilde{h}_{i}(x)(1\leq i\leq \delta(n))$ are irreducible and not self-reciprocal,the number of nonzero coefficients of $\tilde r(x)$ must be odd ,otherwise ${x-1}$ divides $\tilde H(x)$ in $\mathbb{F}_{p}[x]$. It follows that  $\tilde H(x)$ has odd Hamming weight in $ \mathbb{F}_{p}[x]/ {\left\langle x^n-1\right\rangle} $. Computing in $ \mathbb{Z}_{p^m}[x]/ {\left\langle x^n-1\right\rangle }$ , we have 
\begin{equation*}
 F(x) = p^{m/2}H(x) = p^{m/2}[(x^n-1)s(x)+r(x)] = p^{m/2}r(x) = p^{m/2}[\tilde{r}(x)+pr_{1}(x)]\\
\end{equation*}
 for some $r_{1}(x)\in \mathbb{Z}_{p^m}[x]/ {\left\langle {x^n}-1\right\rangle }$ it can be seen that $\omega_{E}(F(x)) = {p^m}(1+pl)$, for some integer $l$. This shows that $\omega_{E}(F(x))$ is not multiple of $p^{m+1}$. This proves that C is of Type $I$ .
\end{proof}
\FloatBarrier
\begin{table}[H]
\centering
\begin{minipage}{0.4\textwidth}
\begin{tabular}{|l|l|l|l|l|l|}\hline
$n$ & $\gamma(n)$ & $\delta(n)$ & $N_{t}$ & $N_{n}$ & $N_{c}$\\\hline
$1$ & $1$ & $0$ & $2$ & $0$ & $1$ \\
$3$ & $2$ & $0$ & $4$ & $0$ & $1$ \\
$5$ & $2$ & $0$ & $4$ & $0$ & $1$ \\
$7$ & $1$ & $1$ & $8$ & $12$ & $4$ \\
$9$ & $3$ & $0$ & $8$ & $0$ & $1$ \\
$11$ & $2$ & $0$ & $4$ & $0$ & $1$ \\
$13$ & $2$ & $0$ & $4$ & $0$ & $1$ \\
$15$ & $3$ & $1$ & $32$ & $48$ & $4$ \\
$17$ & $3$ & $0$ & $8$ & $0$ & $1$ \\
$19$ & $2$ & $0$ & $4$ & $0$ & $1$ \\
$21$ & $2$ & $2$ & $64$ & $336$ & $16$ \\
$23$ & $1$ & $1$ & $8$ & $12$ & $4$ \\
$25$ & $3$ & $0$ & $8$ & $0$ & $1$ \\
$27$ & $2$ & $0$ & $16$ & $0$ & $1$ \\
$29$ & $2$ & $0$ & $4$ & $0$ & $1$ \\
$31$ & $1$ & $3$ & $128$ & $1872$ & $64$ \\
$33$ & $5$ & $0$ & $32$ & $0$ & $1$ \\
$35$ & $2$ & $2$ & $64$ & $336$ & $16$ \\
$37$ & $2$ & $0$ & $4$ & $0$ & $1$ \\
$39$ & $3$ & $1$ & $32$ & $48$ & $4$ \\
$41$ & $3$ & $0$ & $8$ & $0$ & $1$ \\
$43$ & $4$ & $0$ & $16$ & $0$ & $1$ \\
$45$ & $4$ & $2$ & $256$ & $1344$ & $16$ \\
$47$ & $1$ & $1$ & $8$ & $12$ & $4$ \\
$49$ & $1$ & $2$ & $32$ & $168$ & $16$ \\
\hline
\end{tabular}
\FloatBarrier
%\caption{This is a very very long caption which over writes the text on the right side of the paper.}
\label{tab:accuracy} 
\end{minipage}
\hfill
\begin{minipage}{0.55\textwidth}
\centering
\begin{tabular}{|l|l|l|l|l|l|}\hline
$n$ & $\gamma(n)$ & $\delta(n)$ & $N_{t}$ & $N_{n}$ & $N_{c}$ \\
\hline
$51$ & $4$ & $2$ & $256$ & $1344$ & $16$ \\
$53$ & $2$ & $0$ & $4$ & $0$ & $1$ \\
$55$ & $3$ & $1$ & $32$ & $48$ & $4$ \\
$57$ & $5$ & $0$ & $32$ & $0$ & $1$ \\
$59$ & $2$ & $0$ & $4$ & $0$ & $1$ \\
$61$ & $2$ & $0$ & $4$ & $0$ & $1$ \\
$63$ & $3$ & $5$ & $8192$ & $791808$ & $1024$\\ 
$65$ & $7$ & $0$ & $128$ & $0$ & $1$ \\
$67$ & $2$ & $0$ & $4$ & $0$ & $1$ \\
$69$ & $2$ & $2$ & $64$ & $336$ & $16$ \\
$71$ & $1$ & $1$ & $8$ & $12$ & $4$ \\
$73$ & $1$ & $4$ & $512$ & $19488$ & $256$\\ 
$75$ & $4$ & $2$ & $256$ & $1344$ & $16$ \\
$77$ & $2$ & $2$ & $64$ & $336$ & $16$\\ 
$79$ & $1$ & $1$ & $8$ & $12$ & $4$ \\
$81$ & $5$ & $0$ & $32$ & $0$ & $1$ \\
$83$ & $2$ & $0$ & $4$ & $0$ & $1$ \\
$85$ & $4$ & $4$ & $4096$ & $155904$ & $256$ \\
$87$ & $3$ & $1$ & $32$ & $48$ & $4$ \\
$89$ & $1$ & $4$ & $512$ & $19488$ & $256$\\ 
$91$ & $2$ & $4$ & $1024$ & $38976$ & $256$ \\
$93$ & $2$ & $6$ & $16384$ & $170240$ & $729$ \\
$95$ & $3$ & $1$ & $32$ & $48$ & $4$ \\
$97$ & $3$ & $0$ & $8$ & $0$ & $1$ \\
$99$ & $8$ & $0$ & $256$ & $0$ & $1$\\
\hline
\end{tabular}
%\caption{Speed up for the parallel solution of the trivial problem, 16
%Threads on Dual Xeon E-2690.} 
\label{tab:ompdiff} 
\end{minipage}
\end{table}
\par Table above lists the number of cyclic self-orthogonal codes and cyclic self-dual codes over $ \mathbb{Z}_{8} $ for odd length up to $99$ , where $N_{t}$,$N_{n}$
and $N_{c}$ denote the number of trivial cyclic self-orthogonal codes ,non-trivial cyclic self-orthogonal codes and cyclic self-dual codes,respectively.
 Theorem $5.3$ points out that cyclic self-dual codes over $ \mathbb{Z}_{p^m} $ of odd length are always of Type $I$.
\section{conclusion}
\par In this paper we have given cyclic self-orthogonal codes over $ \mathbb{Z}_{p^m} $ of odd length by calculating generator polynomial . Here we have given a  necessary and sufficient condition for the existence of nontrivial cyclic self-orthogonal codes over $ \mathbb{Z}_{p^m} $ . We have determined the enumerator of cyclic self-orthogonal codes over $ \mathbb{Z}_{p^m} $ for a fixed odd length . Cyclic self-dual codes over $ \mathbb{Z}_{p^m} $ of odd length are type $I$ have shown . An interesting problem is to study cyclic self-orthogonal codes over $ \mathbb{Z}_{p^m} $ of odd length when characteristic of ring divide the length of code .
\bibliographystyle{amsplain}

\end{document}